\documentclass[reqno,12pt]{amsart}
\usepackage{amssymb,amsthm,amsmath}
\usepackage{color}
\usepackage{graphicx}
\usepackage{tikz}
\usepackage{comment}
\usepackage{amsfonts}
\usepackage{enumitem}
\usepackage{marginnote}

\usepackage[colorlinks, 
linkcolor=red,
anchorcolor=magenta,
citecolor=blue
]{hyperref}
\usepackage[margin=1.5in]{geometry}

\title[Discrepancy and Quantum Dynamics]{Semi-algebraic discrepancy estimates for multi-frequency shift sequences with applications to quantum dynamics}

\author[]{Wencai Liu}
\address[Wencai Liu]{
	Department of Mathematics, Texas A\&M University, College Station, TX, 77843, USA.
}
\email{liuwencai1226@gmail.com, wencail@tamu.edu}

\author[]{Matthew Powell}
\address[Matthew Powell]{
	Department of Mathematics, Georgia Institute of Technology, Atlanta, GA, 30332, USA.
}
\email{powell@math.gatech.edu}

\author[]{Yiding Max Tang}
\address[Yiding Max Tang]{
        John Burroughs School, St. Louis, MO, 63124, USA.
}
\email{yidingmaxtang@gmail.com, 2028.mtang@jburroughs.org}

\author[]{Xueyin Wang}
\address[Xueyin Wang]{
	Department of Mathematics, Texas A\&M University, College Station, TX, 77843, USA.
}
\email{xueyin@tamu.edu}

\author[]{Ruixiang Zhang}
\address[Ruixiang Zhang]{
    Department of Mathematics, University of California, Berkeley, California 94720, USA
}
\email{ruixiang@berkeley.edu}

\author[]{Justin Zhou}
\address[Justin Zhou]{
        Sperreng Middle School, St. Louis, MO, 63128, USA.
}
\email{justinzhou365@gmail.com}

\newcommand{\Z}{\mathbb{Z}}
\newcommand{\N}{\mathbb{N}}
\newcommand{\T}{\mathbb{T}}
\newcommand{\R}{\mathbb{R}}

\newcommand{\C}{\mathbb{C}}

\theoremstyle{plain}
\newtheorem{theorem}{Theorem}[section]
\newtheorem{corollary}[theorem]{Corollary}
\newtheorem{lemma}[theorem]{Lemma}
\newtheorem{proposition}[theorem]{Proposition}

\theoremstyle{definition}
\newtheorem{definition}[theorem]{Definition}

\newtheorem{remark}[theorem]{Remark}

\DeclareMathOperator{\leb}{Leb}
\DeclareMathOperator{\supp}{supp}
\DeclareMathOperator{\dist}{dist}
\DeclareMathOperator{\spn}{span}

\begin{document}

\begin{abstract}
	We establish asymptotically sharp semi-algebraic discrepancy estimates for multi-frequency  shift sequences. As an application, we obtain novel upper bounds for the quantum dynamics of long-range quasi-periodic Schr\"odinger operators.
\end{abstract}

\maketitle	

\section{Introduction}
In this paper, we are interested in semi-algebraic discrepancy bounds for multi-frequency shift sequences $\{\theta+n\alpha\bmod\mathbb{Z}^{b}\}_{n\in\mathbb{N}}$. More specifically, we will consider a semi-algebraic set $\mathcal S\subseteq [0,1]^{b}$ with $\deg (\mathcal{S})=B$ (see Definition \ref{def:SA} for details). We will fix $\theta, \alpha \in \T^b$ and estimate
\begin{equation}\label{eq:SADisc}
\# \{0 \leqslant  n \leqslant N: \theta + n\alpha \bmod \mathbb{Z}^{b} \in \mathcal S\}.
\end{equation}
Classically, the discrepancy of a sequence $\{x_n\}_{n\in\mathbb{N}} \subseteq \T^b$ is
\begin{equation*}
    D_N(x_n) := \sup_{I \in \mathcal C} \bigg|\frac{\#\{1 \leqslant n \leqslant N: x_n \in I\}}{N} - \leb (I)\bigg|,
\end{equation*}
where $\mathcal C$ denotes the family of all axis-aligned cubes in $[0,1]^b$.
In light of the classical definition, we call \eqref{eq:SADisc} the \textbf{semi-algebraic discrepancy} of the sequence $\{\theta + n\alpha\bmod\mathbb{Z}^{b}\}_{n \in \N}$ since we consider more general sets than cubes. 

As is typical in the study of discrepancy, we will give careful consideration to the arithmetic properties of the frequency vector $\alpha$. To facilitate this, we will introduce two classes of frequencies.

We say that $\alpha \in \mathbb{T}^b$ is \emph{Diophantine} with parameters $\gamma > 0$ and $\tau \geqslant b$, and write $\alpha \in DC(\gamma, \tau)$, if
\begin{equation*}
    \|\langle n, \alpha \rangle\|_{\mathbb{T}} \geqslant \frac{\gamma}{\|n\|_\infty^{\tau}} \quad \text{for all } n \in \mathbb{Z}^b \setminus \{0\},
\end{equation*}
where $\|n\|_\infty = \max_{1 \leqslant j \leqslant b} |n_j|$, $\|x\|_{\mathbb{T}} = \operatorname{dist}(x, \mathbb{Z})$, and  $\langle \cdot, \cdot\rangle$ is the usual $\R^b$ scalar product. It is well-known that for any $\tau > b$, the set
\begin{equation*}
    DC(\tau) := \bigcup_{\gamma > 0} DC(\gamma, \tau)
\end{equation*}
has full Lebesgue measure in $\mathbb{T}^b$.

We also define the notion of a \emph{weak Diophantine} frequency. A vector $\alpha \in \mathbb{T}^b$ is said to be weakly Diophantine with parameters $\gamma > 0$ and $\tau \geqslant 1/b$, denoted by $\alpha \in WDC(\gamma, \tau)$, if
\begin{equation*}
    \|n\alpha\|_{\mathbb{T}^b} \geqslant \frac{\gamma}{|n|^{\tau}} \quad \text{for all } n \in \mathbb{Z} \setminus \{0\},
\end{equation*}
where $\|x\|_{\mathbb{T}^b} := \operatorname{dist}(x, \mathbb{Z}^b) = \max_{1 \leqslant j \leqslant b} \|x_j\|_{\mathbb{T}}$. 
It is known that, for any $\tau > 1/b$, the set
\begin{equation*}
    WDC(\tau) := \bigcup_{\gamma > 0} WDC(\gamma, \tau)
\end{equation*}
also has full Lebesgue measure in $\mathbb{T}^{b}$. 

Before proceeding, let us comment on our use of the two classes $DC$ and $WDC$. In previous works in the spectral theory literature, the large deviation theorem, which yields semi-algebraic exceptional sets used to obtain quantum dynamics, was established through harmonic analysis and crucially relies on $\alpha \in DC$ (see \cite{Bou05, BGLocal, LiuPDE}). 
On the other hand, estimates on classical discrepancy can be derived using the dynamical structure of Kronecker sequences $\{n\alpha\}$. In this context, it is more convenient for us to use $WDC$ \cite{hughbook, Khintchine}. The condition $WDC$ can also be viewed as a natural extension of properties of continued fraction expansions to vectors in $\T^b$.

Our focus will be on properties of so-called semi-algebraic sets.

\begin{definition}[{\cite[Chapter 9]{Bou05}}]\label{def:SA}
	We say $\mathcal{S}\subseteq \mathbb{R}^{b}$ is a semi-algebraic set if it is a finite union of sets defined by a finite number of polynomial inequalities. More precisely, let $\{P_{1}, P_{2}, \cdots, P_{s}\}$ be a family of real polynomials to the variables $x=(x_{1}, x_{2}, \cdots, x_{b})$ with $\deg(P_{i})\leqslant d$ for $i=1,2,\cdots, s$. A (closed) semi-algebraic set $\mathcal{S}$ is given by the expression
	\begin{equation}\label{sas}
		\mathcal{S}=\bigcup_{j} \bigcap_{\ell\in \mathcal{L}_{j}} \{x\in \mathbb{R}^{b}: P_{\ell}(x) \ \varsigma_{j\ell} \ 0\},
	\end{equation}
	where $\mathcal{L}_{j}\subseteq \{1,2,\cdots, s\}$ and $\varsigma_{j\ell}\in \{\geqslant, \leqslant, =\}$. Then we say that the degree of $\mathcal{S}$, denoted by $\deg(\mathcal{S})$, is at most $sd$. In fact, $\deg(\mathcal{S})$ means the smallest $sd$ overall representation as in \eqref{sas}.
\end{definition}

Throughout the paper, we write $A \lesssim_{*} B$ to indicate that there exists a constant $C > 0$ depending on $*$ such that $A \leqslant C B$. For any $x \in \mathbb{R}$, we will also denote by $[x]$ the integer part and by $\{x\}$ the fractional part of $x$. Throughout, we will use the convention that $0 \leqslant \{x\} < 1.$ For $x = (x_1, \dots, x_b) \in \mathbb{R}^b$, we set  
\begin{equation*}
    \{x\} := (\{x_1\}, \dots, \{x_b\}).
\end{equation*}

Our first result is an upper bound on the semi-algebraic discrepancy.

\begin{theorem}\label{aethm}
	Let  $\alpha\in WDC(\gamma, \tau)$. Let $\mathcal{S}\subseteq [0,1]^{b}$ be a semi-algebraic set of degree $B$ and $\leb(\mathcal{S})\leqslant\eta$. Let $N$ be an integer such that 
	\begin{equation*}
		\log N< \frac{1}{2\tau b}\log \frac{1}{\eta}.
	\end{equation*}
	Then for every $\theta \in \mathbb{T}^{b}$, 
	\begin{equation*}
		\#\{1\leqslant n\leqslant N: \theta+n\alpha \bmod \mathbb{Z}^{b}\in \mathcal{S}\} \lesssim_{\gamma,\tau,b} B^{C(b)} N^{\tau(b-1)}.
	\end{equation*}
\end{theorem}

Taking $\tau = 1/b + \varepsilon$ and $\eta\to 0$ in Theorem \ref{aethm}, one quickly obtains 
\begin{equation}\label{eq:asymUpperBound}
    \# \{1\leqslant n\leqslant N: \theta+n\alpha\bmod\mathbb{Z}^{b}\in\mathcal{S}\}\leqslant N^{\frac{b-1}{b}+\varepsilon}.
\end{equation}
We contrast this upper bound with a nearly identical lower bound. In fact, our obtained lower bound agrees with \eqref{eq:asymUpperBound} in the $b\to \infty$ limit.

\begin{theorem}\label{lowbound}
   For almost every $\alpha\in \mathbb{T}^{b}$, the following holds.  For any small $\varepsilon>0$, there exists  $N_{0}=N_{0}(\alpha,\varepsilon)$ such that for any $N\geqslant N_{0}$ there exists a hyperplane $\mathcal{S}\subseteq [0,1]^{b}$ such that
   \begin{equation*}
       \#\{1\leqslant n\leqslant N:n\alpha\bmod\mathbb{Z}^{b}\in\mathcal{S}\} \geqslant N^{\frac{b-1}{b+1}-\varepsilon}.
   \end{equation*}
\end{theorem}

\begin{remark}
   Every hyperplane $\mathcal{S}\subseteq [0,1]^{b}$   is a semi-algebraic set with $\deg (\mathcal{S})\leqslant 2b+1$.
\end{remark}

In light of Theorem \ref{lowbound}, we can see that the upper bound in Theorem \ref{aethm} is \emph{asymptotically sharp in dimension $b$}, in the sense that the ratio between the exponents of upper bound and lower bound tends to 1 as $b \to \infty$:
\begin{equation*}
    \lim_{b \to \infty} \frac{(b-1)/b}{(b-1)/(b+1)}=1.
\end{equation*}

It should be noted that Theorem \ref{aethm} (and particularly \eqref{eq:asymUpperBound}) is, to the best of our knowledge, an improvement on previous best-known results. In particular, for a.e. $\alpha$, \cite{HanJitomirskaya} obtained a bound of $N^{1 - \frac{1}{b^2(b-1) + b} + \varepsilon}$; \cite{JPo} obtained a bound of $N^{1 - \frac{1}{2b} + \varepsilon}$; and \cite{LiuPDE} obtained a bound of $N^{1 - \frac{1}{b^2} + \varepsilon}$. Notably, we have improved the exponent to $1 - \frac 1 b+\varepsilon$. By considering $WDC$, instead of $DC$ as in \cite{HanJitomirskaya, JPo, LiuPDE}, we are able to simultaneously improve the discrepancy estimate and streamline the proof.

As a consequence of our main results, we obtain an upper bound on wave-packet spreading under Schr\"odinger dynamics, otherwise known as quantum dynamics. 
More specifically, we consider a discrete long-range quasi-periodic Schr\"odinger operator $H_{\theta,\alpha}:\ell^2(\Z)\rightarrow \ell^{2}(\mathbb{Z})$
\begin{equation}\label{eq:LongRangeDef}
    \left(H_{\theta,\alpha}\psi\right) (n) = \sum_{m \in \Z} A(n,m) \psi(m) + \lambda V(\theta + n\alpha) \psi(n),
\end{equation}
where $V$ is a real analytic function on the $b$-dimensional torus $\T^b$ (the space of such functions will be denoted $C^\omega(\T^b,\mathbb{R})$), $\theta\in \T^b$ is a phase, $\alpha \in \T^b$ is a frequency, 
$\lambda > 0$ is a coupling constant, and $A(n,m)$ satisfies
\begin{align}
    A(n,m) &= \overline{A(m,n)},\label{eq:SelfAdjoint}\\
    |A(n,m)| &\leqslant C_{1}\mathrm{e}^{-c_{1}|n - m|}, \label{eq:Bounded}\\
    A(n,m)&=A(n+k,m+k),\ k\in\mathbb{Z},\label{eq:ShiftInvariant}
\end{align}
Note that \eqref{eq:SelfAdjoint} and \eqref{eq:Bounded} imply that $H_{\theta,\alpha}$ is both bounded and self-adjoint while \eqref{eq:ShiftInvariant} ensures that $H_{\theta,\alpha}$ is a Toeplitz matrix. 
Moreover, if $A(\cdot, \cdot)$ is the discrete Laplacian, i.e.,
\begin{equation*}
    A(n,m) = \begin{cases}
        0, &|n - m| \neq 1,\\
        1, &|n - m| = 1,
    \end{cases}
\end{equation*}
then \eqref{eq:LongRangeDef} reduces to the classical quasi-periodic Schr\"odinger operator on $\ell^2(\Z).$

Quasi-periodic models, such as the one we consider here, have been studied extensively. The property of particular interest in this paper is quantum dynamics (see \cite{BJBand, Fillman,Fil21, MR4651730, GeKachkovskiy, JitomirskayaZhang,  LiuQD23, ZhangZhao, Zhao, HanJitomirskaya, JPo, ShamisSodinQD,APDL,MR3339185}).

More precisely, consider $H_{\theta,\alpha}$ as the Hamiltonian of a quantum system, and let $\langle \cdot, \cdot \rangle$ denote the inner product on $\ell^2(\mathbb{Z})$. Starting from a compactly supported initial state $\psi \in \ell^2(\mathbb{Z})$ with $\|\psi\|_{2} = 1$, the system evolves under Schr\"odinger dynamics as
\begin{equation*}
    \psi_t(n) = \langle \mathrm{e}^{-\mathrm{i}tH_{\theta,\alpha}} \psi,\delta_{n}\rangle.
\end{equation*}
One of the primary questions of interest is the quantum dynamics of $H_{\theta,\alpha}$: how does $\psi_t$ behave as a function of time $t$ and space $n$? Since $\supp(H_{\theta,\alpha}\psi)$ is typically not compact (at the very least, the support will grow as long as $A(\cdot,\cdot)$ is non-trivial), $\supp(\psi_t)$ will typically grow. The Schr\"odinger evolution, 
$\mathrm{e}^{-\mathrm{i}tH_{\theta,\alpha}},$ 
however, 
is unitary, so $\|\psi_t\|_2 = 1$ for all $t$. 
After this observation, we see that the wave-packet $\psi_t$ ``spreads out'' over time $t$. One way to measure how fast this spreading occurs is to consider where the ``bulk'' of $\psi_t$ is located. That is, there is a $k_t\in \N$ such that 
\begin{equation*}
    \|\psi_t \chi_{[-k_t,k_t]}\|_2 \geqslant 3/4
\end{equation*}
while
\begin{equation*}
    \|\psi_t\chi_{[-k_t+1, k_t-1]}\|_2 < 3/4,
\end{equation*}
where $\chi_{I}$ is the characteristic function of $I$. How $k_t$ behaves, as a function of $t$, is essentially captured by the moments of the position operator
\begin{equation}\label{eq:pthMoment}
    \langle| X_{H_{\theta,\alpha}}|_\psi^p\rangle(t) = \sum_{n \in \mathbb{Z}} |n|^{p} |\langle\mathrm{e}^{-\mathrm{i}tH_{\theta,\alpha}}\psi,\delta_n\rangle|^2,
\end{equation}
see \cite{DamanikTcherem1, DamanikTcherem2} for more details. Since
\begin{equation*}
    d\nu_t := \sum_{n \in \mathbb{Z}} |\langle\mathrm{e}^{-\mathrm{i}tH_{\theta,\alpha}}\psi,\delta_n\rangle |^2 d\delta_n
\end{equation*}
is a probability measure, \eqref{eq:pthMoment} can be interpreted as the $p$-th moment of this probability measure, and thus measures how thin the tail of the measure is (or where the ``bulk'' of the measure is). The $p$-th moment can also be interpreted as the norm of a suitable power of the position operator, which is the origin of the notation $|X|^p$.

It is well-known that the growth of \eqref{eq:pthMoment} is sensitive to the phase $\theta$, which can often be traced to resonant behavior. It is possible, however, to circumvent this sensitivity by averaging over $\theta$. Indeed, it has been shown that the average of \eqref{eq:pthMoment} (in $\theta$) remains finite (known as dynamical localization in expectation) for a.e. $\alpha$ (see, e.g. \cite{Bou05, MR4637128, MR4862298, APDL,JLM24,JKL20}). However, dynamical localization does not hold for every $\theta$ (and in fact does not hold generically) \cite{JitoSimonSC}, so uniform in $\theta$ estimates without averaging are of particular interest.

As an application of Theorem \ref{aethm}, we obtain the following (uniform in $\theta$).
\begin{theorem}\label{mainthm3}
    Let $\alpha\in WDC(\tau) \cap DC(\tau')$. Suppose $V\in C^{\omega}(\mathbb{T}^{b},\mathbb{R})$ is non-constant. 
    Then there exists $\lambda_{0}=\lambda_{0}(A,\alpha,V)>0$ such that for any $\lambda>\lambda_{0}$, $\varepsilon>0$,  $p>0$ and $\psi\in \ell^{2}(\mathbb{Z})$ with compact support, there is $T_{0}=T_{0}(A,\alpha,V,\varepsilon,p,\psi,b)>0$ such that for any $T\geqslant T_{0}$,
    \begin{equation*}
        \sup_{\theta\in\mathbb{T}^{b}} \langle |X_{H_{\theta,\alpha}}|_{\psi}^{p} \rangle (T) \leqslant  (\log T)^{\frac{p}{1-\tau(b-1)}+\varepsilon}.
    \end{equation*}
\end{theorem}

This will follow from Theorem \ref{aethm}. It remains unclear to us whether Theorem \ref{mainthm3} is sharp, though if improvement is possible, other methods will have to be used, as Theorem \ref{aethm} shows that Theorem \ref{mainthm3} exhibits the correct asymptotic behavior in dimension $b$. 

As a consequence of our improvements in Theorem \ref{aethm}, our quantum dynamical bound in Theorem \ref{mainthm3} is also an improvement over the corresponding bounds in \cite{HanJitomirskaya, JPo, LiuPDE}.

Since $WDC(\tau)\cap DC(\tau')$ has full Lebesgue measure for $\tau>1/b$ and $\tau'>b$, we have the following immediate corollary.
\begin{corollary}
    For a.e. $\alpha\in\mathbb{T}^{d}$, the following holds. Suppose $V\in C^{\omega}(\mathbb{T}^{b},\mathbb{R})$ is non-constant. 
    Then there exists $\lambda_{0}=\lambda_{0}(A,\alpha,V)>0$ such that for any $\lambda>\lambda_{0}$, $\varepsilon>0$,  $p>0$ and $\psi\in \ell^{2}(\mathbb{Z})$ with compact support, there is $T_{0}=T_{0}(A,\alpha,V,\varepsilon,p,\psi,b)>0$ such that for any $T\geqslant T_{0}$,
    \begin{equation*}
        \sup_{\theta\in\mathbb{T}^{b}} \langle |X_{H_{\theta,\alpha}}|_{\psi}^{p} \rangle (T) \leqslant  (\log T)^{pb+\varepsilon}.
    \end{equation*}
\end{corollary}

The remainder of our paper is organized as follows. In Section \ref{Sec:UpperBd} we prove Theorem \ref{aethm}. In Section \ref{Sec:LowerBd} we prove Theorem \ref{lowbound}. In Section \ref{Sec:QD} we discuss the connection between quantum dynamics and semi-algebraic discrepancy and prove Theorem \ref{mainthm3}.

\section{Upper bound of discrepancy}\label{Sec:UpperBd}

Before proving Theorem \ref{aethm}, we recall a key property of semi-algebraic sets which will play a critical role in the proof: semi-algebraic sets may be ``covered nicely''. 
More precisely, we have the following lemma, which can be found in \cite{Bou05}.\footnote{In \cite{Bou05}, it is stated that Lemma \ref{covers} follows from the Yomdin-Gromov triangulation theorem \cite{MR880035, MR889980}. For the history and the complete proof of the Yomdin-Gromov triangulation theorem, see \cite{MR3990603}.}

\begin{lemma}[{\cite[Corollary 9.6]{Bou05}}\label{covers}]
	Let $\mathcal{S}\subseteq [0,1]^{b}$ be a semi-algebraic set of degree $B$. Let $\epsilon>0$ be a small number and $\leb(\mathcal{S})\leqslant \epsilon^{b}$. Then $\mathcal{S}$ can be covered by a family of $\epsilon$-balls with total number less than $B^{C(b)} \epsilon^{1-b}$.
\end{lemma}

\subsection{Proof of Theorem \ref{aethm}}

\begin{proof}
    Set $\epsilon = \gamma / (2N^{\tau})$. For sufficiently large $N > N_0(\gamma,\tau,b)$, we have
    \begin{equation*}
        \leb(\mathcal{S}) \leqslant \eta \leqslant N^{-2\tau b} \leqslant \left( \frac{\gamma}{2} \right)^b N^{-\tau b} = \epsilon^b.
    \end{equation*}
    By Lemma~\ref{covers}, the set $\mathcal{S}$ can be covered by at most $B^{C(b)} \epsilon^{1 - b}$ many $\epsilon$-balls.

    We claim that for each such $\epsilon$-ball $D$, there exists at most one integer $n \in [1, N]$ such that 
    \begin{equation*}
        \theta + n\alpha \bmod \mathbb{Z}^b \in D.
    \end{equation*}
    Suppose to the contrary that there exist $1 \leqslant n < n' \leqslant N$ such that
    \begin{equation*}
        \theta + n\alpha \bmod \mathbb{Z}^b\in D, \quad \theta + n'\alpha \bmod \mathbb{Z}^b \in D.
    \end{equation*}
    Then we must have 
    \begin{equation*}
        \|(n - n')\alpha\|_{\mathbb{T}^b} \leqslant 2\epsilon.
    \end{equation*}
    Since $\alpha \in WDC(\gamma, \tau)$, it follows that
    \begin{equation*}
        \|(n - n')\alpha\|_{\mathbb{T}^b} \geqslant \frac{\gamma}{|n - n'|^{\tau}} \geqslant \frac{\gamma}{(N-1)^{\tau}} > 2\epsilon,
    \end{equation*}
    yielding a contradiction. Therefore, there is at most one such $n$ per ball.

    Hence, the number of $n$ satisfying $\theta + n\alpha \bmod \mathbb{Z}^b \in \mathcal{S}$ is bounded by the number of $\epsilon$-balls, which is
    \begin{equation*}
        \# \left\{ 1 \leqslant n \leqslant N : \theta + n\alpha \bmod \mathbb{Z}^{b} \in \mathcal{S} \right\} 
        \leqslant B^{C(b)} \epsilon^{1 - b} \lesssim_{\gamma,b} B^{C(b)} N^{\tau(b - 1)},
    \end{equation*}
    completing the proof.
\end{proof}

\section{Lower bound of discrepancy}\label{Sec:LowerBd}

In this section we will prove Theorem \ref{lowbound} by constructing a special hyperplane of codimension 1 for which it is ``easy'' to compute the semi-algebraic discrepancy. 
\subsection{Construction of independent vectors}

We begin by constructing $b - 1$ vectors to span our hyperplane. Before stating the result, we recall a classical lemma from number theory which will play a key role in our construction.

	\begin{lemma}[\cite{MR159802}]\label{schdis}
		There exists a full measure subset $\Omega\subseteq \mathbb{T}^{b}$ such that for any $\alpha\in \Omega$, there exists $C(\alpha)>0$ such that
		\begin{equation*}
			D_{N}(\{n\alpha\})\leqslant C(\alpha) N^{-1} (\log N)^{b+2}.
		\end{equation*}
	\end{lemma}

\begin{proposition}\label{plane}
	There exists a full measure subset $\Omega\subseteq \mathbb{T}^{b}$ such that for any $\alpha\in \Omega$, the following holds. For any small $\varepsilon>0$ and sufficiently large $N$, there exist $\{n_{1}, \cdots, n_{b-1}\}\subseteq [1,N^{1+2\varepsilon}] \cap \mathbb{Z}$ such that for any $i=1,\cdots, b-1$,
	\begin{equation*}
		\{n_{i}\alpha\}\in [0,1]^{b}\ \text{and}\ \|\{n_{i}\alpha\}\|_\infty \leqslant N^{-1/b},
	\end{equation*}
	and
	\begin{equation*}
		\dim \spn \{ \{n_{1}\alpha\}, \cdots, \{n_{b-1}\alpha\} \}=b-1.
	\end{equation*}
\end{proposition}

\begin{proof}
	
	We will construct $n_1, n_{2},\cdots, n_{b-1}$ by induction.
	
	{\bf Base step:} We begin by constructing $n_{1}$ as follows. Choose any vector $w=(w_{1},\cdots, w_{b})\in [0,1]^{b}$ such that $\|w\|_\infty =\frac{1}{2} N^{-1/b}$. Let $I_{w}\subseteq [0,1]^{b}$ be
	\begin{equation*}
		I_{w}=\prod_{j=1}^{b}\bigg([w_{j}-N^{-\frac{1+\varepsilon}{b}}, w_{j}+N^{-\frac{1+\varepsilon}{b}}]\bigcap [0,1]\bigg).
	\end{equation*}
    Thus $\leb(I_{w})\geqslant N^{-1-\varepsilon}$.
    
	By Lemma \ref{schdis}, for $\alpha\in\Omega$ and sufficiently large $N$,
	\begin{equation*}
		\bigg|\frac{\#\{1\leqslant n\leqslant N^{1+2\varepsilon}: n\alpha \bmod \mathbb{Z}^{b}\in I_{w}\}}{N^{1+2\varepsilon}}-\leb(I_{w})\bigg|\leqslant C(\alpha,b) \frac{(\log N)^{b+2}}{N^{1+2\varepsilon}}.
	\end{equation*}
	Thus
	\begin{equation*}
		\begin{split}
			\#\{1\leqslant n\leqslant N^{1+2\varepsilon}&:  n\alpha \bmod\mathbb{Z}^{b}\in I_{w}\}\\
            &\geqslant N^{1+2\varepsilon}( \leb(I_{w})- C(\alpha,b) N^{-1-2\varepsilon}(\log N)^{b+2})\\
			&\gtrsim N^{\varepsilon},
		\end{split}
	\end{equation*}
	which means that there exists
	$1\leqslant n_{1}\leqslant N^{1+2\varepsilon}$ such that $\{n_{1}\alpha\}\in I_{w}$, thus
	\begin{equation*}
		\|\{n_{1}\alpha\}-w\|_\infty \leqslant N^{-\frac{1+\varepsilon}{b}}.
	\end{equation*}
    Hence
	\begin{equation*}
		\|\{n_{1}\alpha\}\|_\infty \leqslant \|w\|_\infty+ N^{-\frac{1+\varepsilon}{b}} \leqslant N^{-1/b}.
	\end{equation*}
	Moreover, it is obvious that
	\begin{equation*}
		\dim \spn \{\{n_{1}\alpha\}\}=1.
	\end{equation*}
	This completes the base case of our induction.

	{\bf Induction step: } We will now assume that we have $n_{1},n_{2}, \cdots, n_{k-1}$ such that for any $i=1,\cdots, k-1$,
	\begin{equation}\label{a1}
		\{n_{i}\alpha\}\in [0,1]^{b} \ \text{and}\ \|\{n_{i}\alpha\}\|_\infty \leqslant N^{-1/b},
	\end{equation}
	and
	\begin{equation}\label{a2}
		\dim \spn\left\{\{n_{1}\alpha\},\cdots, \{n_{k-1}\alpha\}\right\}=k-1.
	\end{equation}
	We will construct $n_{k}$ such that \eqref{a1} and \eqref{a2} both hold when $k-1$ is replaced by $k$.

    We will proceed as follows. For any two non-zero vectors $u$ and $v$ in $\mathbb{R}^{b}$, we define
	\begin{equation*}
		\dist(u,v)=\sqrt{1-\frac{|\langle u,v\rangle|^{2}}{\langle u,u\rangle \langle v,v\rangle}},
	\end{equation*}
	The function $\dist(\cdot,\cdot)$ is the natural angular distance. Define $\|u\|_{2}=\sqrt{\langle u,u\rangle}$.

    Choose any $w=(w_{1},\cdots, w_{b})\in [0,1]^{b}$ such that
    \begin{equation}\label{wup}
        \|w\|_{2}=\frac{1}{2}N^{-1/b}.
    \end{equation}
    and, for some $\delta > 0$,
    \begin{equation}\label{wv}
        \dist(w,v)\geqslant \delta
    \end{equation}
    for any $v\in \spn \{\{n_{1}\alpha\}, \cdots, \{n_{k-1}\alpha\}\}$.
    Define $I_{w}\subseteq [0,1]^{b}$ as the cube
    \begin{equation*}
        I_{w}=\prod_{j=1}^{b}\bigg([w_{j}-N^{-\frac{1+\varepsilon}{b}}, w_{j}+N^{-\frac{1+\varepsilon}{b}}]\bigcap [0,1]\bigg).
    \end{equation*}
    Obviously, we have $\leb(I_{w})\geqslant N^{-1-\varepsilon}$. 

    By Lemma \ref{schdis}, for $\alpha\in\Omega$ and sufficiently large $N$,
    \begin{equation*}
        D_{N^{1+2\varepsilon}}(\{n\alpha\}) \leqslant C(\alpha) \frac{(\log N^{1+2\varepsilon})^{b+2}}{N^{1+2\varepsilon}}\leqslant C(\alpha,b)\frac{(\log N)^{b+2}}{N^{1+2\varepsilon}}.
    \end{equation*}
	Thus for sufficiently large $N$, we have
	\begin{equation*}
		\begin{split}
			\#\{1\leqslant n\leqslant N^{1+2\varepsilon}&:  n\alpha \bmod\mathbb{Z}^{b}\in I_{w}\}\\
            &\geqslant N^{1+2\varepsilon}( \leb(I_{w})- C(\alpha,b) N^{-1-2\varepsilon}(\log N)^{b+2})\\
			&\gtrsim N^{\varepsilon},
		\end{split}
	\end{equation*}
    which means that there exists
	$1\leqslant n_{k}\leqslant N^{1+2\varepsilon}$ such that $\{n_{k}\alpha\}\in I_{w}$, thus
    \begin{equation}\label{close}
		\|\{n_{k}\alpha\}-w\|_\infty \leqslant N^{-\frac{1+\varepsilon}{b}}.
	\end{equation}
	Combining \eqref{wup} with \eqref{close} yields
	\begin{equation}\label{r1}
		\|\{n_{k}\alpha\}\|_\infty \leqslant \|w\|_{\infty}+N^{-\frac{1+\varepsilon}{b}}\leqslant  N^{-1/b}.
	\end{equation}
    This shows $n_k$ satisfies \eqref{a1}. It remains to show $n_k$ satisfies \eqref{a2}.

    For any $v\in \spn\{\{n_{1}\alpha\}, \cdots, \{n_{k-1}\alpha\}\}$, by \eqref{wv}, we have
	\begin{equation}\label{d1}
		\begin{split}
			\dist(\{n_{k}\alpha\}, v) &\geqslant \dist(w,v)-\dist(w, \{n_{k}\alpha\})\\
			&\geqslant \delta- \sqrt{1-\frac{|\langle w, \{n_{k}\alpha\} \rangle|^{2}}{\langle w,w\rangle \langle \{n_{k}\alpha\},\{n_{k}\alpha\}\rangle}}.
		\end{split}
	\end{equation}
	It follows from \eqref{wup} and \eqref{close} that for sufficiently large $N$, depending on $\varepsilon, \delta,$ and $b,$
	\begin{equation}\label{d2}
		\begin{split}
			\frac{|\langle w, \{n_{k}\alpha\} \rangle|}{ \sqrt{\langle w,w\rangle \langle \{n_{k}\alpha\},\{n_{k}\alpha\}\rangle } } 
            &\geqslant \frac{\|w\|_{2}^{2}-|\langle w, \{n_{k}\alpha\} -w\rangle|}{ \|w\|_{2} \cdot \|\{n_{k}\alpha\}\|_{2} } \\
            &\geqslant \frac{\|w\|_{2}-\|\{n_{k}\alpha\} -w\|_{2} }{\|\{n_{k}\alpha\}\|_{2}}\\
            &\geqslant \frac{\|w\|_{2}-\|\{n_{k}\alpha\} -w\|_{2}}{\|w\|_{2}+\|\{n_{k}\alpha\} -w\|_{2}}\\
            &\geqslant \frac{N^{-\frac{1}{b}}-\sqrt{b}N^{-\frac{1+\varepsilon}{b}}}{N^{-\frac{1}{b}}+\sqrt{b}N^{-\frac{1+\varepsilon}{b}}}\\
			&\geqslant 1-\delta^{10}.
		\end{split}
	\end{equation}
	By \eqref{d1} and \eqref{d2}, we have
	\begin{equation*}
		\dist(\{n_{k}\alpha\}, v) \geqslant \frac{\delta}{2}>0,
	\end{equation*}
	which means 
	\begin{equation*}
		\{n_{k}\alpha\} \notin \spn\{\{n_{1}\alpha\},\cdots, \{n_{k-1}\alpha\}\}.
	\end{equation*}
    Thus
	\begin{equation}\label{r2}
		\dim \spn\{\{n_{1}\alpha\},\cdots, \{n_{k}\alpha\}\}=k.
	\end{equation}
	Combing \eqref{r1} with \eqref{r2} shows $n_k$ satisfies \eqref{a2}, which completes the proof of the induction step.
\end{proof}

\subsection{Proof of Theorem \ref{lowbound}}

In this subsection we will give a lower bound on the discrepancy for the hyperplane constructed above. 

\begin{theorem}\label{count}
	There exists a full measure subset $\Omega\subseteq \mathbb{T}^{b}$ such that for any $\alpha \in \Omega$, the following holds. For any small $\varepsilon>0$ and sufficiently large $N$, there exists a hyperplane $\mathcal S\subseteq [0,1]^{b}$ such that
	\begin{equation*}
		\# \{1\leqslant n\leqslant N: n\alpha\bmod \mathbb{Z}^{b}\in \mathcal{S}\}\geqslant  N^{ \frac{b-1}{b + 1}-\varepsilon}.
	\end{equation*}
\end{theorem}
\begin{proof}
	Let $r=\frac{b}{b+1}$. By Proposition \ref{plane}, for any small $\varepsilon>0$ and sufficiently large $N$, there exist $\{n_{1},\cdots, n_{b-1}\}\subseteq [1, N^{r(1+2\varepsilon)}]\cap \mathbb{Z}$ such that for any $i=1,\cdots, b-1$,
	\begin{equation*}
		\{n_{i}\alpha\}\in[0,1]^{b}\ \text{and}\ \| \{n_{i}\alpha\}\|_\infty\leqslant N^{-r/b},
	\end{equation*}
	and
	\begin{equation}\label{gjun252}
		\dim \spn \{\{n_{1}\alpha\}, \cdots, \{n_{b-1}\alpha\} \}=b-1.
	\end{equation}
	    
	Define the hyperplane $\mathcal S$ as 
	\begin{equation*}
		\mathcal{S} := [0,1]^{b}\bigcap \spn \{ \{n_{1}\alpha\}, \{n_{2}\alpha\},\cdots, \{n_{b-1}\alpha\}  \}. 
	\end{equation*}
	It is clear that
    for any positive integers $k_i$ with  $i=1,2,\cdots, b-1$ satisfying 
    \begin{equation*}
        1\leqslant k_{i}n_{i}\leqslant \frac{N}{b} \ \text{and}\ k_{i}   \|\{n_{i}\alpha\}\|_{\infty}<\frac{1}{b}, 
    \end{equation*}
    we have
    \begin{equation}\label{gjun251}
		(k_1\{n_1 \alpha\} + \cdots + k_{b-1}\{n_{b-1} \alpha\} ) \in \mathcal{S}.
	\end{equation}
    Therefore,
	\begin{equation*}
		\begin{split}
			\# \{1\leqslant n\leqslant N& :  n\alpha\bmod\mathbb{Z}^{b}\in \mathcal{S}\}\\
            &= \# \{n\alpha\bmod\mathbb{Z}^{b}: 1\leqslant n\leqslant N,n\alpha\bmod\mathbb{Z}^{b}\in \mathcal{S}\}\\
            \text{ by  \eqref{gjun251}}&\geqslant \# \{(\sum_{i=1}^{b-1}k_{i} n_i\alpha)\bmod\Z^b: 1\leqslant k_{i}n_{i}\leqslant \frac{N }{b}\ \text{and}\ k_{i}   \|\{n_{i}\alpha\}\|_{\infty}<\frac{1}{b}\}\\
            &= \# \{\sum_{i=1}^{b-1}k_{i} \{n_i\alpha\}: 1\leqslant k_{i}n_{i}\leqslant \frac{N }{b}\ \text{and}\ k_{i}   \|\{n_{i}\alpha\}\|_{\infty}<\frac{1}{b}\}\\
			\text{ by  \eqref{gjun252}}&=\prod_{i=1}^{b-1}\# \{k_{i}: 1\leqslant k_{i}n_{i}\leqslant \frac{N}{b} \ \text{and}\ k_{i}   \|\{n_{i}\alpha\}\|_{\infty}<\frac{1}{b}\}\\
			&\geqslant \prod_{i=1}^{b-1} \min \{b^{-1}N^{1-r(1+2\varepsilon)}, b^{-1}N^{r/b}\}\\
			&\geqslant  N^{\frac{b-1}{b+1}-\varepsilon},
		\end{split}
	\end{equation*}
	which completes our proof.
\end{proof}

\section{Quantum dynamics}\label{Sec:QD}

Before proceeding with the proof of Theorem \ref{mainthm3}, we recall a few preliminary notions which connect quantum dynamics to semi-algebraic discrepancy.

\subsection{A large deviation theorem}
The following mirrors the discussion from \cite{LiuQD23} specialized to dimension 1. Let $H_{\theta,\alpha}$ be given by \eqref{eq:LongRangeDef}. In this section, we will fix $K$ sufficiently large so that $\sigma(H_{\theta,\alpha}) \subseteq [-K+1, K-1]$.

First, we define the Green's function of $H_{\theta,\alpha}$. Let  $R_{{\Lambda}}$  be the operator of restriction (i.e. projection) to $\Lambda \subseteq \Z$. The Green's function (restricted to $\Lambda$) at $z$ is defined as
\begin{equation*}\label{g0}
	G_{{\Lambda}}(z,\theta)=(R_{{\Lambda}}(H_{\theta,\alpha}-zI)R_{{\Lambda}})^{-1}.
\end{equation*}
Clearly, $G_{{\Lambda}}(z,\theta)$ is always well-defined for $z\in \C_+\equiv \{z\in \C: \Im z>0\}$ by self-adjointness of $H_{\theta,\alpha}$.

We say the Green's function  of an operator $H_{\theta,\alpha}$  satisfies the Large Deviation Theorem (LDT) in complexified energies if there exist $\epsilon_0>0$ and $N_0>0$ such that for  any  $N\geqslant N_0$, there exists a  subset $\Theta_N\subseteq \mathbb{T}^b$  such that
\begin{equation*}
	\leb(\Theta_N)\leqslant \mathrm{e}^{-{N}^{\sigma_2}},
\end{equation*}
and for any $\theta\notin  \Theta_N \bmod \Z^b$,
\begin{equation}\label{sg}
    \begin{split}
        \| G_{[-N,N]}(z,\theta) \|   & \leqslant \mathrm{e}^{N^{\sigma_1}},                                       \\
	|G_{[-N,N]}(z,\theta)(n,m)| & \leqslant \mathrm{e}^{-c_2 |n-m|}, \text{ for } |n-m|\geqslant N/10,
    \end{split}
\end{equation}
where  $z=E+\mathrm{i}\epsilon$ with $E\in[-K+1,K-1]$ and $0<\epsilon\leqslant \epsilon_0$.

The following result from \cite{LiuPDE} shows that the LDT holds for $H_{\theta,\alpha}$ given by \eqref{eq:LongRangeDef}.

\begin{theorem}[{\cite[Theorem 3.11]{LiuPDE}}\label{LDTshift}]
    Let $H_{\theta,\alpha}$ be given by \eqref{eq:LongRangeDef}. Let $\alpha\in DC(\gamma,\tau)$ and $1-1/(b\tau)<\sigma<1$. Then for any $\varepsilon>0$, there exists $\lambda_{0}=\lambda_{0}(A,\alpha,V,\sigma,\varepsilon)>0$ such that for any $\lambda>\lambda_{0}$ and any $N$, there exists $\Theta_{N}\subseteq \mathbb{T}^{b}$ such that
    \begin{equation*}
	\leb(\Theta_N)\leqslant \mathrm{e}^{-{N}^{\sigma-1/(b^{2}\tau)+1/(b^{3}\tau^{2})-\varepsilon}},
\end{equation*}
and for any $\theta\notin  \Theta_N \bmod \Z^b$,
\begin{align}
	\| G_{[-N,N]}(z,\theta) \|   & \leqslant \mathrm{e}^{N^{\sigma}},     \label{g07174}                               \\
	|G_{[-N,N]}(z,\theta)(n,m)| & \leqslant \mathrm{e}^{-\frac{1}{2}c_{1} |n-m|}, \ \text{for}\  |n-m|\geqslant N/10.\label{g07171}
\end{align}
\end{theorem}

The LDT and semi-algebraic discrepancy were explicitly linked to quantum dynamics in \cite{LiuQD23}.

\begin{theorem}[{\cite[Corollary 2.3]{LiuQD23}}]\label{criterion}
	Define $\mathcal{B}_{N,N_1}$ as
	\begin{equation*}
		\mathcal{B}_{N,N_1}=\{ n\in [-N,N]: G_{[-N_{1},N_{1}]}(z,\theta+n\alpha) \ \text{doesn't satisfy \eqref{sg}}\}.
	\end{equation*}
	Assume that there exists $\epsilon_0>0$ and $N_{0}>0$ such that for any $z=E+\mathrm{i}\epsilon $ with  $|E|\leqslant K$ and $0<\epsilon\leqslant \epsilon_0$, and arbitrarily  small $\varepsilon>0$,
	\begin{equation*}
		\#  \mathcal{B}_{N,[ N^{\varepsilon} ]}\leqslant N^{1-\delta} {\text{ when }} N \geqslant N_0 
	\end{equation*}
	($N_0$ may depend on $\varepsilon$).
	Then for any $\phi$ with compact support and any $\varepsilon>0$ there exists $T_0>0$ (depending on $b, p,\phi,K,\sigma_{1},\delta, \epsilon_0, c_1, c_{2}, C_1, N_0$ and $\varepsilon$) such that for any $T\geqslant T_0$,
	\begin{align*}
		\langle|{X}_{H_{\theta,\alpha}}|_{\phi}^p\rangle(T)        \leqslant  (\log T)^{  \frac{p}{\delta} +\varepsilon}. 
	\end{align*}
\end{theorem}

\subsection{Proof of Theorem \ref{mainthm3}}
Using Theorems \ref{LDTshift} and \ref{criterion} as our starting point, we may now establish Theorem \ref{mainthm3}.

\begin{proof}
Let $\varepsilon>0$ be sufficiently small and let $N_{1}=[N^{\varepsilon}]$ be sufficiently large (depending on $\varepsilon$).
By Theorem \ref{LDTshift}, for $\alpha\in DC(\tau')$ and $\sigma=1-1/(2b\tau')$ there exists $\lambda_{0}=\lambda_{0}(A,\alpha,V)>0$ such that for every $\lambda>\lambda_{0}$ the LDT holds with
\begin{equation}\label{g07172}
    \leb(\Theta_{N_{1}})\leqslant \mathrm{e}^{-N_{1}^{c}}<\mathrm{e}^{-N^{\varepsilon}}.
\end{equation}
By first approximating the analytic potential with trigonometric polynomials, and then further approximating the trigonometric functions using their Taylor series, one can, via standard perturbation theory (see p.56 of \cite{Bou05}), assume that $\Theta_{N_1}$ is a semi-algebraic set of degree at most 
\begin{equation}\label{g07173}
    \deg(\Theta_{N_1}) \leqslant N_1^{C} < N^{\varepsilon}.
\end{equation}

By \eqref{g07172}, \eqref{g07173} and Theorem \ref{aethm}, we know that for $\alpha\in WDC(\tau)$ and sufficiently large $N$,
\begin{equation}\label{sublinearlast}
    \# \{n\in [-N,N]: \theta+n\alpha \bmod\mathbb{Z}^{b}\in \Theta_{N_{1}}\} \leqslant N^{\tau(b-1)+\varepsilon}.
\end{equation}
Theorem \ref{mainthm3} now follows immediately from Theorem \ref{criterion} and \eqref{sublinearlast}.
\end{proof}

\section*{Acknowledgments}

This project was conducted as part of the  High School Research Program ``PReMa'' (Program for Research in Mathematics) at Texas A\&M University.  We thank the program director, Kun Wang, for bringing the initial team together.

W. Liu was a Simons Fellow in Mathematics for the 2024–2025 academic year. He gratefully acknowledges the Department of Mathematics at the University of California, Berkeley for their hospitality during his visit, where part of this work was completed. R. Zhang is supported by NSF CAREER DMS-2143989 and Sloan Research Fellowship.  

This work was supported in part by NSF DMS-2246031,  DMS-2052572,  and the Dolciani Mathematics Enrichment Grant from MAA.

\section*{Statements and Declarations}
{\bf Conflict of Interest} 
The authors declare no conflicts of interest.

\vspace{0.2in}
{\bf Data Availability}
Data sharing is not applicable to this article as no new data were created or analyzed in this study.

\bibliographystyle{amsalpha}
\bibliography{main}

\providecommand{\bysame}{\leavevmode\hbox to3em{\hrulefill}\thinspace}
\providecommand{\MR}{\relax\ifhmode\unskip\space\fi MR }
\providecommand{\MRhref}[2]{%
  \href{http://www.ams.org/mathscinet-getitem?mr=#1}{#2}
}
\providecommand{\href}[2]{#2}
\begin{thebibliography}{CLSW25}

\bibitem[BG00]{BGLocal}
J.~Bourgain and M.~Goldstein, \emph{On nonperturbative localization with quasi-periodic potential}, Ann. of Math. (2) \textbf{152} (2000), no.~3, 835--879. \MR{1815703}

\bibitem[BJ00]{BJBand}
J.~Bourgain and S.~Jitomirskaya, \emph{Anderson localization for the band model}, Geometric aspects of functional analysis, Lecture Notes in Math., vol. 1745, Springer, Berlin, 2000, pp.~67--79. \MR{1796713}

\bibitem[BN19]{MR3990603}
G.~Binyamini and D.~Novikov, \emph{Complex cellular structures}, Ann. of Math. (2) \textbf{190} (2019), no.~1, 145--248. \MR{3990603}

\bibitem[Bou05]{Bou05}
J.~Bourgain, \emph{Green's function estimates for lattice {S}chr\"odinger operators and applications}, Annals of Mathematics Studies, vol. 158, Princeton University Press, Princeton, NJ, 2005. \MR{2100420}

\bibitem[CLSW25]{APDL}
A.~Cai, H.~Lv, Y.~Shan, and X.~Wang, \emph{Arithmetic $(k-d/2)$-polynomial dynamical localization for $k$ power-law quasi-periodic long-range operators on $\mathbb{Z}^d$}, preprint (2025).

\bibitem[DLY15]{MR3339185}
D.~Damanik, M.~Lukic, and W.~Yessen, \emph{Quantum dynamics of periodic and limit-periodic {J}acobi and block {J}acobi matrices with applications to some quantum many body problems}, Comm. Math. Phys. \textbf{337} (2015), no.~3, 1535--1561. \MR{3339185}

\bibitem[DT07]{DamanikTcherem1}
D.~Damanik and S.~Tcheremchantsev, \emph{Upper bounds in quantum dynamics}, J. Amer. Math. Soc. \textbf{20} (2007), no.~3, 799--827. \MR{2291919}

\bibitem[DT08]{DamanikTcherem2}
\bysame, \emph{Quantum dynamics via complex analysis methods: general upper bounds without time-averaging and tight lower bounds for the strongly coupled {F}ibonacci {H}amiltonian}, J. Funct. Anal. \textbf{255} (2008), no.~10, 2872--2887. \MR{2464194}

\bibitem[Fil17]{Fillman}
J.~Fillman, \emph{Ballistic transport for limit-periodic {J}acobi matrices with applications to quantum many-body problems}, Comm. Math. Phys. \textbf{350} (2017), no.~3, 1275--1297. \MR{3607475}

\bibitem[Fil21]{Fil21}
\bysame, \emph{Ballistic transport for periodic {J}acobi operators on {$\mathbb{ Z}^d$}}, From operator theory to orthogonal polynomials, combinatorics, and number theory---a volume in honor of {L}ance {L}ittlejohn's 70th birthday, Oper. Theory Adv. Appl., vol. 285, Birkh\"auser/Springer, Cham, [2021] \copyright 2021, pp.~57--68. \MR{4367462}

\bibitem[GK23]{GeKachkovskiy}
L.~Ge and I.~Kachkovskiy, \emph{Ballistic transport for one-dimensional quasiperiodic {S}chr\"{o}dinger operators}, Comm. Pure Appl. Math. \textbf{76} (2023), no.~10, 2577--2612. \MR{4630598}

\bibitem[Gro87]{MR880035}
M.~Gromov, \emph{Entropy, homology and semialgebraic geometry}, Ast\'erisque (1987), no.~145-146, 5, 225--240, S\'eminaire Bourbaki, Vol.\ 1985/86. \MR{880035}

\bibitem[GYZ23]{MR4637128}
L.~Ge, J.~You, and Q.~Zhou, \emph{Exponential dynamical localization: criterion and applications}, Ann. Sci. \'Ec. Norm. Sup\'er. (4) \textbf{56} (2023), no.~1, 91--126. \MR{4637128}

\bibitem[HJ19]{HanJitomirskaya}
R.~Han and S.~Jitomirskaya, \emph{Quantum dynamical bounds for ergodic potentials with underlying dynamics of zero topological entropy}, Anal. PDE \textbf{12} (2019), no.~4, 867--902. \MR{3869380}

\bibitem[JKL20]{JKL20}
S.~Jitomirskaya, H.~Kr\"uger, and W.~Liu, \emph{Exact dynamical decay rate for the almost {M}athieu operator}, Math. Res. Lett. \textbf{27} (2020), no.~3, 789--808. \MR{4216568}

\bibitem[JLM24]{JLM24}
S.~Jitomirskaya, W.~Liu, and L.~Mi, \emph{Sharp palindromic criterion for semi-uniform dynamical localization}, arXiv preprint arXiv:2410.21700 (2024).

\bibitem[JP22]{JPo}
S.~Jitomirskaya and M.~Powell, \emph{Logarithmic quantum dynamical bounds for arithmetically defined ergodic {S}chr{\"o}dinger operators with smooth potentials}, Analysis at Large: Dedicated to the Life and Work of Jean Bourgain, Springer, 2022, pp.~173--201.

\bibitem[JS94]{JitoSimonSC}
S.~Jitomirskaya and B.~Simon, \emph{Operators with singular continuous spectrum. {III}. {A}lmost periodic {S}chr\"odinger operators}, Comm. Math. Phys. \textbf{165} (1994), no.~1, 201--205. \MR{1298948}

\bibitem[JZ22]{JitomirskayaZhang}
S.~Jitomirskaya and S.~Zhang, \emph{Quantitative continuity of singular continuous spectral measures and arithmetic criteria for quasiperiodic {S}chr\"{o}dinger operators}, J. Eur. Math. Soc. (JEMS) \textbf{24} (2022), no.~5, 1723--1767. \MR{4404788}

\bibitem[Khi64]{Khintchine}
A.~Ya. Khinchin, \emph{Continued fractions. ({U}niversity of {C}hicago {P}ress)}, The Mathematical Gazette \textbf{49} (1964), no.~369.

\bibitem[Liu22]{LiuPDE}
W.~Liu, \emph{Quantitative inductive estimates for {G}reen's functions of non-self-adjoint matrices}, Anal. PDE \textbf{15} (2022), no.~8, 2061--2108. \MR{4546503}

\bibitem[Liu23]{LiuQD23}
\bysame, \emph{Power law logarithmic bounds of moments for long range operators in arbitrary dimension}, J. Math. Phys. \textbf{64} (2023), no.~3, Paper No. 033508, 11. \MR{4564259}

\bibitem[Mon94]{hughbook}
H.~L. Montgomery, \emph{Ten lectures on the interface between analytic number theory and harmonic analysis}, CBMS Regional Conference Series in Mathematics, vol.~84, Published for the Conference Board of the Mathematical Sciences, Washington, DC; by the American Mathematical Society, Providence, RI, 1994. \MR{1297543}

\bibitem[Sch64]{MR159802}
W.~M. Schmidt, \emph{Metrical theorems on fractional parts of sequences}, Trans. Amer. Math. Soc. \textbf{110} (1964), 493--518. \MR{159802}

\bibitem[Sha25]{MR4862298}
Y.~Shan, \emph{Dynamical localization for finitely differentiable quasi-periodic long-range operators}, J. Differential Equations \textbf{427} (2025), 803--826. \MR{4862298}

\bibitem[SS23]{ShamisSodinQD}
M.~Shamis and S.~Sodin, \emph{Upper bounds on quantum dynamics in arbitrary dimension}, J. Funct. Anal. \textbf{285} (2023), no.~7, Paper No. 110034, 20. \MR{4604835}

\bibitem[Yom87]{MR889980}
Y.~Yomdin, \emph{{$C^k$}-resolution of semialgebraic mappings. {A}ddendum to: ``{V}olume growth and entropy''}, Israel J. Math. \textbf{57} (1987), no.~3, 301--317. \MR{889980}

\bibitem[You23]{MR4651730}
G.~Young, \emph{Ballistic transport for limit-periodic {S}chr\"odinger operators in one dimension}, J. Spectr. Theory \textbf{13} (2023), no.~2, 451--489. \MR{4651730}

\bibitem[Zha17]{Zhao}
Z.~Zhao, \emph{Ballistic transport in one-dimensional quasi-periodic continuous {S}chr\"{o}dinger equation}, J. Differ. Equ. \textbf{262} (2017), no.~9, 4523--4566. \MR{3608164}

\bibitem[ZZ17]{ZhangZhao}
Z.~Zhang and Z.~Zhao, \emph{Ballistic transport and absolute continuity of one-frequency {S}chr\"{o}dinger operators}, Comm. Math. Phys. \textbf{351} (2017), no.~3, 877--921. \MR{3623240}

\end{thebibliography}
\end{document}